\documentclass[conference]{IEEEtran}
\IEEEoverridecommandlockouts
\usepackage{cite}
\usepackage{amsmath,amssymb,amsfonts}
\usepackage{graphicx}
\usepackage{textcomp}
\usepackage{xcolor}
\usepackage{algpseudocode}
\usepackage{algorithm}
\usepackage{multirow}
\usepackage{float}
\usepackage{subcaption}
\usepackage{amsthm}
\usepackage{afterpage}
\usepackage{graphicx}
\usepackage{amsmath}
\usepackage{booktabs}
\newtheorem{theorem}{Theorem}
\algnewcommand\algorithmicforeach{\textbf{for each}}
\algdef{S}[FOR]{ForEach}[1]{\algorithmicforeach\ #1\ \algorithmicdo}
\def\BibTeX{{\rm B\kern-.05em{\sc i\kern-.025em b}\kern-.08em
    T\kern-.1667em\lower.7ex\hbox{E}\kern-.125emX}}
\begin{document}

\title{Generator Cost Coefficients Inference Attack via Exploitation of Locational Marginal Prices in Smart Grid\\}

	\author{Junfei Wang, {\em {Member}},~IEEE,
	Pirathayini Srikantha, {\em {Member}},~IEEE
	\thanks{J. Wang and P. Srikantha are with the Department of Electrical Engineering and Computer Science, York University, Toronto, ON, Canada.; E-mails: junfeiw@yorku.ca and psrikan@yorku.ca. }
}

\maketitle

\begin{abstract}
Real-time price signals and power generation levels (disaggregated or aggregated) are commonly made available to the public by Independent System Operators (ISOs) to promote efficiency and transparency. However, they may inadvertently reveal crucial private information about the power grid, such as the cost functions of generators. Adversaries can exploit these vulnerabilities for strategic bidding, potentially leading to financial losses for power market participants and consumers. In this paper, we prove the existence of a closed-form solution for recovering coefficients in cost functions when LMPs and disaggregated power generation data are available. Additionally, we establish the convergence conditions for inference the quadratic coefficients of cost functions when LMPs and aggregated generation data are given. Our theoretical analysis provides the conditions under which the algorithm is guaranteed to converge, and our experiments demonstrate the efficacy of this method on IEEE benchmark systems, including 14-bus and 30-bus and 118-bus systems.
\end{abstract}


\section{Introduction}\label{sec:intro}
To promote transparency, monitor emissions, and help consumers and market participants make more informed decisions, Independent System Operators (ISOs) and governmental entities such as the Environmental Protection Agency (EPA) may publish open market data, including prices and power generation levels. For example, Independent System Operator New England (ISO-NE) \cite{isone} provides a real-time Locational Marginal Price (LMP) map and aggregates power generation data by resource type every five minutes. LMPs, which represent real-time price signals at different locations in the market, are typically derived as the Lagrange multipliers of the AC Optimal Power Flow (ACOPF) problem, which is solved at fixed intervals (e.g. every five minutes \cite{Cain2012}). As outputs of OPF, both power generation levels and LMPs are highly valuable, and research in ACOPF often incorporates them to facilitate fast decision-making processes \cite{Liu2022, Wang2022}. Although this information is not inherently vulnerable, adversaries may exploit the dataset to infer critical private details about the power grid, such as generators' cost functions.

There is growing concern recently regarding cybersecurity risks associated with leveraging LMPs for attacks on power system operations. Research on LMP-related cybersecurity has shown that publicly available LMPs can be used to conduct stealthy False Data Injection Attacks (FDIAs) against state estimation algorithms \cite{Asghari2023}, and that FDIAs can also reveal changes in LMPs \cite{Zhang2020} in the power grids. In \cite{Kekatos2014}, historical spatiotemporal price data is exploited for power grid topology tracking under both noisy and noiseless conditions. Additionally, LMPs are used to infer electricity market information such as transmission capacity, loss parameters and active constraints in \cite{Birge2017}. Paper \cite{Gu2022} proposes an online approach for recovering the complete grid topology from timely LMP data and incomplete partial topology information. Recovering cost functions of generators as attack mechanism has attracted growing attention in the community of cybersecurity in power system \cite{Donti2018, Chen2019}. Cost functions of power generators are crucial in OPF and power bidding. If the coefficients of these functions are exposed to adversaries or competitors, they may engage in strategic bidding, potentially causing significant financial losses for other generators and consumers \cite{David2000}. Reference \cite{Donti2018} assumes power demand and LMP at each bus and generation level at each generator are available, and proposes a machine learning-based algorithm to recover both admittance matrix and cost functions. However, experiment shows that no unique solution is guaranteed, and initial guesses as well as weights in the loss function may result in incorrect solutions. In \cite{Chen2019}, the cost function is assumed to be piece-wise linear and  discovered via a modified random search algorithm considering historical market clearing prices and power generations levels. Nevertheless, the uniqueness of the solution in this method may not be guaranteed.

In this paper, we show that adversaries can accurately recover the coefficients of each generator's cost function in power grids by combining two key pieces of public information—power generation levels and LMPs—with a small amount of data from the targeted bus. This work proves the existence of closed-form solutions of both quadratic and linear coefficients in the scenario where both power generation levels and LMPs at all nodes are available. Furthermore, when LMPs and only aggregated power generation data are given, the proposed method theoretically establishes the conditions under which the algorithm can converge to a unique solution of quadratic coefficients. Experimental results show that the cost function recovery attack method converges to the ground truth values on the IEEE 14-bus, 30-bus and 118-bus systems when the specified conditions are met.

The remainder of this paper is organized as follows. Sec~\ref{sec:form} mathematically formulates the ACOPF problem and the LMPs, and then introduces two scenarios considered in the paper with different threat models (i.e. assumptions). Then in Sec~\ref{sec:method}, the attack mechanisms in these scenarios are proposed. Subsequently, the experimental settings and the efficacy of the proposed algorithms is demonstrated in Sec~\ref{sec:ex}. Finally, we summarize the work in this paper and discusse the future extensions of this research in Sec~\ref{sec:conc}.  
\vspace{5pt}
\section{Problem Formulation}\label{sec:form}
In this section, we firstly formulate the ACOPF problem, and then demonstrate the relationship between LMPs, Lagrangian of ACOPF and the Karush-Kuhn-Tucker (KKT) conditions. Finally, the two scenarios with different threat models will be introduced.

\subsection{Formulation of ACOPF}
In power flow analysis, an electrical grid is composed of buses (nodes) and electrical lines (branches). Each bus is associated with four essential quantities: active power injection, denoted as $P_{G_{i}}-P_{D_{i}}$, the difference between active power generation and demand; reactive power generation $Q_{G_{i}}-Q_{D_{i}}$, the difference between reactive power generation and demand; voltage magnitude ($|V_{i}|$), typically maintained within operational limits (e.g., 0.95-1.05 p.u.) to ensure safe operations in the grid; phase angle ($\phi_{i}$), the relative phase position w.r.t the reference bus. The set of all buses in a power grid is denoted as $\mathcal{N}$, and generator bus set is represented by $\mathcal{G}$, where $\mathcal{G} \in \mathcal{N}$. Electrical lines are represented by Nodal Admittance Matrix $Y$, also known as Y-bus matrix, which is a symmetric complex matrix. 

For each generator $g\in \mathcal{G}$, the function $C_{g}(\cdot)$ is the cost associated with power generation, which is typically a convex quadratic function of the active power generation $P_{G_g}$, defined in Eq.~\ref{eq:cost}. 

\begin{equation}\label{eq:cost}
	C_{i}(P_{G_{i}})=a_{i}{P_{G_{i}}}^{2}+b_{i}{P_{G_{i}}}+c_i
\end{equation}

In the cost function, the coefficients $a_i$ and $b_i$ determine the how quickly the marginal generation cost increases at generator $i$. These coefficients are the the most critical parameters used in electricity market bidding, as they directly impact the generators' offer prices. Meanwhile, $c_i$ is the fixed cost associated with generator $i$, incurred even no electricity is generated. 

A typical ACOPF problem can be formulated in $\mathcal{P}_{ACOPF}$. It aims to minimize the total generation cost in the power grid, while satisfying all physical constraints [C1]-[C6]. For each generator, constraints [C1] and [C2] impose capacity limits on active and reactive power dispatch, respectively. The power flow equation in the Bus Injection Model \cite{Low2014}, defined in [C3], ensures the power balance at each bus. Constraints [C4] and [C5] establish upper and lower bounds for voltage magnitudes and phase angles. Finally, the maximum allowable power flow for each electrical line is defined by [C6].

{\footnotesize
	\begin{align*}
		\mathcal{P}_{ACOPF}:  & \min_{P_{G_{i}}} \sum_{i\in\mathcal{G}}C_{i}(P_{G_{i}}) & \\
		\text{s.t.} \  \forall  \ i \in \mathcal{G}: \ &P_{G_{i}}^{min}\leq P_{G_{i}} \leq P_{G_{i}}^{max} & \textbf{[C1]} \\
		&Q_{G_{i}}^{min}\leq Q_{G_{i}} \leq Q_{G_{i}}^{max} & \textbf{[C2]} \\
		\text{s.t.} \  \forall  \ i \in \mathcal{N}: \ & P_{G_{i}}-P_{D_{i}}+j(Q_{G_{i}}-Q_{D_{i}}) =V_{i}\sum_{j\in \mathcal{N}}V_{j}^{*}Y_{ij}^* & \textbf{[C3]}\\
		&|V_i^{min}|\leq |V_{i}| \leq |V_i^{max}| & \textbf{[C4]} \\
		&\phi_i^{min}\leq \phi_{i} \leq \phi_i^{max} & \textbf{[C5]}\\
		&{|V_{i}(V_{i}^{*}-V_{j}^{*})Y_{ij}^*|\leq S_{i,j}^{max} \ \forall \  j \in \mathcal{N},j\ne i }& \textbf{[C6]}
	\end{align*}
}
Due to the non-convexity imposed by the constraints [C3], [C4] and [C6], ACOPF problem is NP-hard \cite{Lavaei2011}. Many of the existing solvers are based on Interior Point Method \cite{Zimmerman2010,Zamzam2020,Pan2021},  leveraging the Lagragian and its Karush-Kuhn-Tucker (KKT) conditions \cite{Wang2022} to iteratively solve approximations of the original problem. The LMPs are derived from the Lagrangian multipliers of the ACOPF. These will be elaborated in the next section. 

\subsection{Locational Marginal Prices}
The Lagrangian of $\mathcal{P}_{ACOPF}$ is defined as follows in Eq.~\ref{eq:lag}, where $C_i$ is the cost function associated with generator $i \in \mathcal{G}$, $\nu_i^+$, $\nu_i^-$, $\zeta_i^+$ and $\zeta_i^-$ are dual variables of constraint [C1] and [C2], $\lambda_{i}$ and $\mu_i$ serves as dual variables in the real and imaginary components for [C3], the dual variables of [C4] - [C6] are $\eta_i^+$, $\eta_i^-$, $\gamma_i^+ $, $\gamma_i^-$ and $\xi_{ij}$, respectively. Furthermore, $\lambda_{i}$ is used as the LMP at each node $i$ in the power market.

{\footnotesize
\begin{equation} \label{eq:lag}
	\begin{split}
		\mathcal{L} = & \sum_{i \in \mathcal{G}} C_i(P_{G_{i}}) +\sum_{i \in \mathcal{G}} [\nu_i^+ \left(P_{G_{i}} - {P_{G_{i}}^{\max}} \right) + \nu_i^- \left({P_{G_{i}}^{\min}} - P_{G_{i}} \right)\\
		&+\zeta_i^+ \left(Q_{G_{i}} - {Q_{G_{i}}^{\max}} \right)+ \zeta_i^- \left({Q_{G_{i}}^{\min}} - Q_{G_{i}} \right)]\\
		&\sum_{i \in \mathcal{N}} [\lambda_i \left(P_{D_{i}}-P_{G_{i}} + \sum_{j} Re(V_{i}\sum_{j\in \mathcal{N}}V_{j}^{*}Y_{ij}^*) \right) \\
		&+ \mu_i \left(Q_{D_{i}} - Q_{G_{i}} + \sum_{j}Im(V_{i}\sum_{j\in \mathcal{N}}V_{j}^{*}Y_{ij}^*) \right) \\
		&+ \eta_i^+ \left(|V_{i}| - |V_{i}^{\max}| \right)+\eta_i^- \left(|V_{i}^{\min}|  - |V_{i}| \right)\\
		&+ \gamma_i^+ \left(\phi_{i} - \phi_{i}^{\max} \right)+\gamma_i^- \left(\phi_{i}^{\min}  - \phi_{i} \right)]\\
		&+ \sum_{ij \in \mathcal{E}} \xi_{ij} \left(|V_{i}(V_{i}^{*}-V_{j}^{*})Y_{ij}^*| - S_{ij}^{\max} \right)
	\end{split}
\end{equation}}

The Karush-Kuhn-Tucker (KKT) conditions are both necessary and sufficient for the optimality of convex optimization problems, but only necessary for the optimality of ACOPF due to its non-convexity. These conditions are given by:

\begin{align*}
	& \text{(1) Stationarity: }  && \nabla \mathcal{L} = 0 \text{ w.r.t all variables}, \\
	& \text{(2) Primal feasibility: }  && [C1]-[C6] \text{  adhered },  \\
	& \text{(3) Dual feasibility: }  && \text{non-negative dual variables} \\
	& &&\text{of inequality constraints, }\\
	& \text{(4) Complementary slackness: }  && \text{all inequality constraint are }  \\
	&&& \text{either active or with zero dual}\\
\end{align*}
Based on the Stationary condition, the gradient $\frac{\partial \mathcal{L} }{\partial P_{i}^{G}}=0, \forall i \in \mathcal{G}$ when $P_{i}^{G}$ is the optimal solution. Moreover, in Complementary slackness, if a constraint is inactive, the associated dual variable must be zero.

\subsection{Assumptions in This Paper}

\begin{enumerate}
	\item \text{Scenario 1: } This scenario has the similar assumptions as in \cite{Donti2018}, where all LMPs and power generations are given, capacity constraints of each generator may be binding, and other constraints are not active.  

	\item \text{Scenario 2: } ISOs may only publicize aggregated generations and LMPs. In this scenario, $\lambda_{i}$ at bus $i$ is coupling of $a_g$ and $\lambda_{g}$ all other buses $\forall g\in\mathcal{G}$, which is highly nonlinear. We assume no binding constraint exist in the dataset.
\end{enumerate}

We will also show in Sec.~\ref{sec:method} that the coefficients $a_{i}$ and $b_{i}$ as well as the capacity information can be inferred by adversaries in Scenario 1, and in Scenario 2 the quadratic coefficient $a_i$ can be recovered.

\vspace{5pt}
\section{Proposed Algorithm}\label{sec:method}
In this section, we propose a closed-form solution for inferring cost coefficients in Scenario 1 and introduce two cases in which the maximum or minimum generation capacity of generators can be estimated. Furthermore, we derive a novel Fixed-Point Iteration Method for recovering quadratic coefficients of cost functions in Scenario 2, where only LMPs and aggregated power generation data are available to adversaries.

\subsection{Attack Mechanism Using LMPs and Individual Generations}\label{sec:3A}
In this scenario, we assume adversaries are capable of collecting a historical operational dataset in the power grid containing pairs of LMP and power generation at each generator bus. Based on the Stationary condition, these pairs satisfy Eq.~\ref{eq:linear}:\\
\begin{equation}\label{eq:linear}
	2a_{i}P_{i}^{G}+b_{i} - \lambda_i + \nu_i^+ - \nu_i^- = 0
\end{equation}

Additionally, following from the Complementary Slackness condition, where the dual variables $\nu_i^+$ and $\nu_i^-$ are both zero when ${P_{G_i}^{\min}}<P_{G_i}<{P_{G_i}^{\max}}$, constraint [C1] in $\mathcal{P}_{ACOPF}$ can be deactivated by excluding data points corresponding to the maximum and minimum values of  $P_{G_i}$. Then, the Stationary condition will be as follows in Eq.~\ref{eq:linear2}, where $\nu_i^+$ and $\nu_i^-$ are both zero. Moreover, this condition holds locally at each generator bus $i\in\mathcal{G}$, so only using local data at bus $i$ is sufficient for inferring the cost coefficients $a_i$ and $b_i$.\\
\begin{equation}\label{eq:linear2}
	2a_{i}P_{G_{i}}+b_{i} - \lambda_i = 0
\end{equation}
Since it is a two dimensional linear problem, only two pairs of $P_{G_{i}}$ and $ \lambda_{i}$ are required to fully recover $a_{i}$ and $b_{i}$. Let's define the two points as $\alpha=\{P_{G_{i}}^\alpha,\lambda_{i}^\alpha\}$ and $\beta=\{P_{G_{i}}^\beta,\lambda_{i}^\beta\}$, then the solution is derived in Eq.~\ref{eq:lrsolu}.\\
\begin{equation}\label{eq:lrsolu}
 \left\{\begin{matrix}
	a_i = \frac{\lambda_{i}^\alpha-\lambda_{i}^\beta}{2(P_{G_{i}}^\alpha-P_{G_{i}}^\beta)}\\
	b_i = \lambda_i^{\alpha}-\frac{P_{G_{i}}^\alpha(\lambda_{i}^\alpha-\lambda_{i}^\beta)}{P_{G_{i}}^\alpha-P_{G_{i}}^\beta}\\
\end{matrix}\right.
\end{equation}

After the cost coefficients are reconstructed, attackers can examine if the the capacity of the target generator can be inferred. This can be achieved by testing whether the extreme values in the dataset along with their LMPs satisfy the linear relationship defined in Eq.~\ref{eq:linear2}. When $P_{i}^{G}$ reaches the bound, either $\nu_i^+$ or $\nu_i^-$ is positive. Therefore, power generation and LMPs will adhere to Eq.~\ref{eq:linear} instead of Eq.~\ref{eq:linear2}, and the two cases below define the information adversaries can determine:
\begin{enumerate}
	\item Case 1: When $P_{G_{i}}<\frac{\lambda_{i}-b_{i}}{2a_{i}}$, it reveals $\nu_i^+>0$ and $\nu_i^-=0$. Then, the upper bound capacity of the generator $i$ is $P_{G_{i}}^{max}=P_{G_{i}}$.
	\item Case 2: Conversely, $\nu_i^->0$ and $\nu_i^+=0$ when $P_{G_{i}}>\frac{\lambda_{i}-b_{i}}{2a_{i}}$, so the lower bound capacity is $P_{G_{i}}^{min}=P_{G_{i}}$. 
\end{enumerate}

\subsection{Attack Algorithm Exploiting LMPs and Aggregated Power Generations}
The real-time aggregated power generation data are often publicized from particular resources ( e.g., hydropower, natual gas, nuclear data in CAISO \cite{Caiso}). The aggregate generation denoted as $P_{G_{A}}=\sum_{g \in \mathcal{G}} P_{G_{g}}$ is known, while each single $P_{G_{g}}$ remains private. This scenario is designed to prevent the leakage of critical information about the power grid in Scenario 1, e.g. historical bidding data that can be leveraged for reverse engineering \cite{Donti2018}. However, by combining LMPs with aggregated power generation data, a Multivariate Fixed Point Iteration (MFPI) method is proposed in this section to infer the quadratic coefficient of each generator's cost function. It is theoretically proven that once vulnerable data points are located, the MFPI algorithm is guaranteed to converge to the unique solution.

We select a base data point $\beta_i$ and an auxiliary data point $\alpha_i$ for each generator $i\in\mathcal{G}$ denoted as $\beta_i=\{P_{G_{A}}^{(\beta_i)}, \lambda_{1}^{(\beta_i)}, \lambda_{2}^{(\beta_i)}, ..., \lambda_{|G|}^{(\beta_i)}\}$ and $\alpha_i=\{P_{G_{A}}^{(\alpha_i)}, \lambda_{1}^{(\alpha_i)}, \lambda_{2}^{(\alpha_i)}, ..., \lambda_{|G|}^{(\alpha_i)}\}$, containing the aggregated power generation in the power grid and LMPs at all generator buses. These two points satisfy the following Eqs.~\ref{eq:x0}:\\
\begin{equation}
	\left\{\begin{matrix}
		2a_{i}(P_{G_{A}}^{(\beta_{i})}-\sum_{g \in \mathcal{G}, g\ne i}\frac{\lambda_{g}^{(\beta_{i})}-b_{g}}{2a_{g}})&+b_{i} - \lambda_{i}^{(\beta_{i})} = 0\\
		2a_{i}(P_{G_{A}}^{(\alpha_{i})}-\sum_{g \in \mathcal{G}, g\ne i}\frac{\lambda_{g}^{(\alpha_{i})}-b_{g}}{2a_{g}})&+b_{i} - \lambda_{i}^{(\alpha_{i})} = 0
	\end{matrix}\right.	
	\label{eq:x0}
\end{equation}
Noting that $\frac{\lambda_{g}^{(\beta_{i})}-b_{g}}{2a_{g}}$ and $\frac{\lambda_{g}^{(\alpha_{i})}-b_{g}}{2a_{g}}$ in the equations are the active power generations $P_{G_{g}}^{(\beta_{i})}$ and $P_{G_{g}}^{(\alpha_{i})}$ at bus $g\in\mathcal{G},g\ne i$. Thus, $P_{G_{A}}^{(\beta_{i})}-\sum_{g \in \mathcal{G}, g\ne i}\frac{\lambda_{g}^{(\beta_{i})}-b_{g}}{2a_{g}}$ and $P_{G_{A}}^{(\alpha_{i})}-\sum_{g \in \mathcal{G}, g\ne i}\frac{\lambda_{g}^{(\alpha_{i})}-b_{g}}{2a_{g}}$ are active power generation at bus $i$ for the two points. Using auxiliary point as an independent equation to substitute and eliminate $b_i$ in the base point, Eqs.~\ref{eq:x0} can be rewritten in Eq.~\ref{eq:ag}:\\
\begin{equation}\label{eq:ag}
	2a_{i}(P_{G_{A}}^{(\beta_i)}-P_{G_{A}}^{(\alpha_i)}+\sum_{g \in \mathcal{G}, g\ne i}\frac{\lambda_{g}^{(\alpha_i)}-\lambda_{g}^{(\beta_i)}}{2a_{g}})+\lambda_{i}^{(\alpha_i)}- \lambda_{i}^{(\beta_i)}=0
\end{equation}
By rearranging Eq.~\ref{eq:ag}, the expression of $a_i$ can be derived in the form of Eq.~\ref{eq:eq8}. Clearly, $a_{i}$ is a function of all other second order coefficients $a_g,g\in \mathcal{G},g\ne i$. \\
\begin{equation}\label{eq:eq8}
	a_{i}=\frac{\lambda_{i}^{(\beta_i)}- \lambda_{i}^{(\alpha_i)}}{2(P_{G_{A}}^{(\beta_i)}-P_{G_{A}}^{(\alpha_i)}+\sum_{g \in \mathcal{G}, g\ne i}\frac{\lambda_{g}^{(\alpha_i)}-\lambda_{g}^{(\beta_i)}}{2a_{g}})}
\end{equation}
Applying this to all generators in the system, the computation of $a_i$ for all generators can be generalized into a system of equations, where the entire set of coefficients can be expressed as a MFPI defined in Eq.~\ref{eq:fa}:
\begin{equation}
	a=\begin{pmatrix}
		a_{1}\\
		a_{2}\\
		\vdots \\
		a_{|\mathcal{G}|}
	\end{pmatrix}= \begin{pmatrix}
		F_{1}(a_{1}, a_{2}, ..., a_{|\mathcal{G}|})\\
		F_{2}(a_{1}, a_{2}, ..., a_{|\mathcal{G}|})\\
		\vdots \\
		F_{|\mathcal{G}|}(a_{1}, a_{2}, ..., a_{|\mathcal{G}|}))
	\end{pmatrix}=F(a)
	\label{eq:fa}
\end{equation}

Both of the existence and uniqueness of MFPI are guaranteed by Banach's fixed point theorem\cite{Su2016} when the following condition holds:\\
\begin{equation}
	d(F(a),F(\hat{a})) \le L d(a,\hat{a}), \forall a, \hat{a} \in \mathcal{A}
\end{equation}    
where $F:\mathcal{A}\to \mathcal{A}$ is a self-map of a metric space $(\mathcal{A}, d)$, $L \in (0,1)$ is the Lipschitz constant, and $a=(a_{1}, a_{2}, ..., a_{|\mathcal{G}|})\in a^{|\mathcal{G}|}$, $\hat{a}=(\hat{a}_{1}, \hat{a}_{2}, ..., \hat{a}_{|\mathcal{G}|})\in a^{|\mathcal{G}|}$. 

We will show that the contractive mapping of $F(a)$ is satisfied by the lower bound of the norm of its Jacobian matrix in Theorem~\ref{tr:convergence}. The Jacobian matrix of $F(a)$ can be defined in Eq.~\ref{eq:jac}:\\
\begin{equation}\label{eq:jac}
	F'(a) = \left[ \begin{array}{ccccc}
		\frac{\partial F_1}{\partial a_1} & \frac{\partial F_1}{\partial a_2} & \cdots & \frac{\partial F_1}{\partial a_{|\mathcal{G}|}} \\
		\frac{\partial F_2}{\partial a_1} & \frac{\partial F_2}{\partial a_2} & \cdots & \frac{\partial F_2}{\partial a_{|\mathcal{G}|}}  \\
		\vdots & \ddots & \ddots & \vdots \\
		\vdots & \ddots & \ddots & \vdots \\
		\frac{\partial F_{|\mathcal{G}|}}{\partial a_1} & \frac{\partial F_{|\mathcal{G}|}}{\partial a_2} & \cdots & \frac{\partial F_{|\mathcal{G}|}}{\partial a_{|\mathcal{G}|}}
	\end{array} \right]. 
\end{equation}
where  all diagonal terms are zero, and off diagonal term can be calculated as:\\
\begin{equation}
	\frac{\partial F_{i}}{\partial a_{g}}= \frac{(\lambda_{i}^{(\beta_{i})}-\lambda_{i}^{(\alpha_{i})})(\lambda_{g}^{(\alpha_{i})}-\lambda_{g}^{(\beta_{i})})}{4a_{g}^{2}(P_{G_{A}}^{(\beta_{i})}-P_{G_{A}}^{(\alpha_{i})}+\sum_{\hat{g} \in \mathcal{G}, \hat{g}\ne i}\frac{\lambda_{\hat{g}}^{(\alpha_{i})}-\lambda_{\hat{g}}^{(\beta_{i})}}{2a_{\hat{g}}})^2}\\
\end{equation}

\begin{theorem}\label{tr:convergence}
	The MFPI attack for recovering cost coefficients via LMPs and aggregated power generations is guaranteed to converge to an unique solution when the following conditions hold:
	\begin{enumerate}
		\item $\lambda_{i}^{(\beta_i)}- \lambda_{i}^{(\alpha_i)}>0, \forall i\in \mathcal{G}$
		\item $P_{G_{A}}^{(\beta_i)}-P_{G_{A}}^{(\alpha_i)}>0,\forall i \in \mathcal{G}$
		\item $a^{max} \ge F_i(a)\ge a^{min}>0, \forall i \in\mathcal{G}$
	\end{enumerate}
where $a_{\min}$ is defined in Eq.~\ref{eq:bound}, and $a_{\max}$ is a predefined upper bound for $a_i$.
	
\end{theorem}
\begin{proof}
In the expression of $a_i$ defined in Eq.~\ref{eq:eq8}, the dominant part of the denominator is $P_{G_{A}}^{(\beta_i)}-P_{G_{A}}^{(\alpha_i)}$, which is the difference of aggregated power generations. Due to to the positivity of the coefficient $a_i$ and the need for deriving bound in Condition 3, both $\lambda_{i}^{(\beta_i)}- \lambda_{i}^{(\alpha_i)}$ and $P_{G_{A}}^{(\beta_i)}-P_{G_{A}}^{(\alpha_i)}$ have to be positive.
	
Furthermore, to prove the convergence of MFPI, we need to show that 
	\begin{equation*}
		\frac{|F(a)-F(\hat{a})|}{|(a-\hat{a})|}\le 1
	\end{equation*}
	Based on the Mean Value Theorem \cite{Flett1958}, there exists at least one data point $\xi$ between $a$ and $\hat{a}$ satisfying:\\
	\begin{equation*}
		F(a)-F(\hat{a})=F'(\xi)(a-\hat{a})
	\end{equation*}
	Taking the matrix norm on both sides and applying the Sub-multiplicative property, the following inequality can be established:\\
	\begin{equation*}
		\frac{|F(a)-F(\hat{a})|}{|(a-\hat{a})|}	\le|F'(\xi)|
	\end{equation*}
	Instead of explicitly searching for the data point $\xi$, if the norm of all possible $\xi$ in the Jacobian matrix $F'(a)$ are bounded by $L \in (0,1)$, the convergence of the proposed fixed point iteration is guaranteed. Due to computational efficiency, $L_\infty$ is adopted. The lower bound $a_{\min}$ in Condition 3 is derived in Eq.\ref{eq:bound}, where Conditions 1 and 2 are applied, and the summation of the fraction term in the denominator is divided into positive and negative case for inducing the upper bound of the norm. To bound the maximum value of the norm, $a_{\min}$ is derived. Additionally, the self-mapping $F(a)\in \mathcal{A}$ has to be guaranteed in Condition 3.
	
\begin{figure*}[t]  
	\begin{equation}\label{eq:bound}
		\begin{split}
			&L_{\infty}(F'(a))=\max_{i} \sum_{g\in\mathcal{G}}|\frac{\partial F_{i}}{\partial a_{g}}|\\
			&=\max_{i}\frac{|\lambda_{i}^{(\beta_{i})}-\lambda_{i}^{(\alpha_{i})}|}{4(P_{G_{A}}^{(\beta_{i})}-P_{G_{A}}^{(\alpha_{i})}+\sum_{\hat{g} \in \mathcal{G}, \hat{g}\ne i}\frac{\lambda_{\hat{g}}^{(\alpha_{i})}-\lambda_{\hat{g}}^{(\beta_{i})}}{2a_{g}})^2}\sum_{g \in \mathcal{G}, g\ne i} \frac{|\lambda_{g}^{(\alpha_{i})}-\lambda_{g}^{(\beta_{i})}|}{a_{g}^{2}}\\
			&\leq \max_{i}\frac{\lambda_{i}^{(\beta_{i})}-\lambda_{i}^{(\alpha_{i})}}{4(P_{G_{A}}^{(\beta_{i})}-P_{G_{A}}^{(\alpha_{i})}+\sum_{\hat{g} \in \mathcal{G}^{+}, \hat{g}\ne i}\frac{\lambda_{\hat{g}}^{(\alpha_{i})}-\lambda_{\hat{g}}^{(\beta_{i})}}{2a^{max}}+\sum_{\hat{g} \in \mathcal{G}^{-}, \hat{g}\ne i}\frac{\lambda_{\hat{g}}^{(\alpha_{i})}-\lambda_{\hat{g}}^{(\beta_{i})}}{2a^{min}})^2}\sum_{g \in \mathcal{G}, g\ne i} \frac{|\lambda_{g}^{(\alpha_{i})}-\lambda_{g}^{(\beta_{i})}|}{a_{g}^{2}}\\
			&\leq \max_{i}\frac{(\lambda_{i}^{(\beta_{i})}-\lambda_{i}^{(\alpha_{i})})\sum_{g \in \mathcal{G}, g\ne i}|\lambda_{g}^{(\alpha_{i})}-\lambda_{g}^{(\beta_{i})}|}{4(P_{G_{A}}^{(\beta_{i})}-P_{G_{A}}^{(\alpha_{i})}+\sum_{\hat{g} \in \mathcal{G}^{+}, \hat{g}\ne i}\frac{\lambda_{\hat{g}}^{(\alpha_{i})}-\lambda_{\hat{g}}^{(\beta_{i})}}{2a^{max}}+\sum_{\hat{g} \in \mathcal{G}^{-}, \hat{g}\ne i}\frac{\lambda_{\hat{g}}^{(\alpha_{i})}-\lambda_{\hat{g}}^{(\beta_{i})}}{2a^{min}})^2} \frac{1}{{a^{min}}^{2}}<1\\
			&\therefore a^{min}\ge\max_{i}\frac{\sqrt{(\lambda_{i}^{(\beta_{i})}-\lambda_{i}^{(\alpha_{i})})\sum_{g \in \mathcal{G}, g\ne i}|\lambda_{g}^{(\alpha_{i})}-\lambda_{g}^{(\beta_{i})}|}-\sum_{\hat{g} \in \mathcal{G}^{-}, \hat{g}\ne i}(\lambda_{\hat{g}}^{(\alpha_{i})}-\lambda_{\hat{g}}^{(\beta_{i})})}{2(P_{G_{A}}^{(\beta_{i})}-P_{G_{A}}^{(\alpha_{i})})+\sum_{\hat{g} \in \mathcal{G}^{+}, \hat{g}\ne i}\frac{\lambda_{\hat{g}}^{(\alpha_{i})}-\lambda_{\hat{g}}^{(\beta_{i})}}{a^{max}}}
		\end{split}
	\end{equation}
\end{figure*}
	
\end{proof}
\vspace{5pt}
\section{Experiment}\label{sec:ex}

In this section, we will first detail the environment used in this research, including the software, libraries, computation platform, etc. Then, we will illustrate the numerical results of the proposed method.

\subsection{Environment and Dataset}
In this work, the experiment ran on a local laptop equipped with an Intel Core i7-10750H CPU and 32 GB of RAM. Python was used to implement the algorithm. To assess the performance of our proposed algorithm, we conducted experiments on three standard benchmark systems: IEEE 14-bus, IEEE 30-bus and IEEE 118-bus systems. The ACOPF dataset was generated by MATPOWER's primal-dual interior point solver (MIPS)\cite{Zimmerman2010}. 

For each of these benchmark systems, a dataset comprising 10K pairs of input and optimal output data was generated. We introduced variability by randomly perturbing the nodal active and reactive power demands. This perturbation ranged from 80\% to 120\% of their respective nominal values, thereby simulating real-world fluctuations in power demand. This is the common practice in the literature  \cite{Liu2022,Pan2021,Wang2024}. 

There are five generators in the IEEE 14-bus system, one of which is offline for most of the time (84.9\%), and another has 271 data points at its minimum capacity. None of the generators reaches full capacity. This light-loaded scenario gives 1,509 points inactive to constraints. In IEEE 30-bus grid, there are 6 generators all with quadratic cost functions. In the dataset, 2 generators with 3,365 data points are at their minimum capacity, 4 generators with 3,997 points reach the full capacity. Among the 10,000 data points, 6,546 are not binding to any constraints associated with active power generation. \textcolor{black}{In IEEE 118-bus system, there are 54 generators, none of which never reach the upper limits, while 15 are almost always offline. All data points have at least one binding constraint.} The cost function in Matpower is used as the ground truth in this experiment, where the quadratic coefficients range from $0.01$ to $0.25$ in two systems.

\subsection{Attack Performance on Scenario 1} 
We use power generation points that are more than 0.1 p.u. away from both the maximum and minimum values in the dataset and randomly select two points with a generation difference of at least 0.01 p.u. to avoid instability. 

Running Eq.~\ref{eq:lrsolu} results in $3.55e^{-4}$ and $3.03e^{-4}$ MSE errors for IEEE 14-bus system on the coefficient $a_i$ and $b_i, i\in\mathcal{G}$, respectively. For IEEE 30-bus system, the recovery error for $a_i$ and $b_i, i\in\mathcal{G}$ are as low as $3.64e^{-12}$ and $1.85e^{-8}$, respectively. \textcolor{black}{These average error on IEEE 118-bus system are $1.32e^{-3}$ and $0.012$.}

Because no point reaches the upper capacity limits in the dataset of IEEE 14-bus system, even after adversaries remove points associated with highest power generation, the algorithm can only infer the lower limits for two of the generators with points adhering Case 2 in Sec.~\ref{sec:3A}. The MSE error is at $1.48e^{-4}$ in average. For 30-bus system, the maximum capacity of four generators out of six can be recovered, where running the algorithm with the points satisfying Case 1 in Sec.~\ref{sec:3A} reach $0.106$ MSE error in average. For the minimum power generation, two of generators data in Case 2 results in $0.052$ MSE error in average. Since no data points reach the maximum generation limits in the IEEE 118-bus system dataset, the algorithm can only infer the lower limits. The lower limits of 24 generators are reconstructed with an MSE error of  $9.56e^{-5}$.

\subsection{Attack Performance on Scenario 2} 
In this section, we conduct an experiment to verify the convergence behaviour of the proposed MFPI algorithm. Since MFPI convergence is guaranteed when the conditions are satisfied, our primary objective is to identify points that meet these conditions and validate their impact on the algorithm’s performance. 

To successfully use MFPI algorithm recovering cost coefficients with LMPs and aggregated power generations, $2|\mathcal{G}|$ points denoted as $\beta_{i}$ and $\alpha_{i}, i\in\mathcal{G}$ satisfying 3 conditions in Theorem.~\ref{tr:convergence} need to be collected. We search across the dataset to verify conditions for random paris of $\beta_{i}$ and $\alpha_{i}$ until all $2|\mathcal{G}|$ points are found. In Theorem.~\ref{tr:convergence}, conditions 1 and 2 can be verified easily in searching of vulnerable points in the dataset.  The maximum value $a_{\max}$ in condition 3, which is used to derived lower bound in Eq.~\ref{eq:bound} in the worst case, is chosen to be 0.5 in the experiment. Condition 3 is not easy to access, because $a_i$ is unknown. To increase the probability for the hold of this condition, the searching algorithm run through the dataset to seek the lowest bound $a_{\min}$. 

We show the MFPI's performance on not only $\pm 20\%$ random perturbations of nominal load as inputs to ACOPF, but also the $\pm 10\%$ and $\pm 50\%$ perturbations in Tab.~\ref{tab:sample}. For both 14-bus and 30-bus systems in all sample ranges, the algorithm converges in 10 iterations with error parameter errors below $1e^-6$.

\begin{table}[htbp]
	\caption{Coefficient Inference Performance}
	\centering
	\begin{tabular}{cccccc}
		\toprule
		\text{Grid} & \text{Sample Range} & \#Iter & $mse_a$ \\
		\midrule
		\multirow{3}{*}{IEEE 14-Bus} & 10\% & 10 & $2.57e^{-7}$ \\
		& 20\% & 7 & $8.29e^{-7}$ \\
		& 50\% & 10 & $5.37e^{-8}$ \\
		\midrule
		\multirow{3}{*}{IEEE 30-Bus} & 10\% & 9 & $5.56e^{-12}$ \\
		& 20\% & 9 & $5.78e^{-10}$ \\
		& 50\% & 8 & $7.02e^{-12}$ \\
		\bottomrule
	\end{tabular}
	\label{tab:sample}
\end{table}

The convergence process of MFPI on IEEE 14-bus with 4 active generators is shown in Fig.~\ref{fig:mfpi}, where the x-axis represents the number of iterations, and y-axis is the difference between predicted values and the ground truth. The errors on all generators decrease nearly zero after the third iteration. 
\begin{figure}[tb]
	\centerline{\includegraphics[scale=0.7]{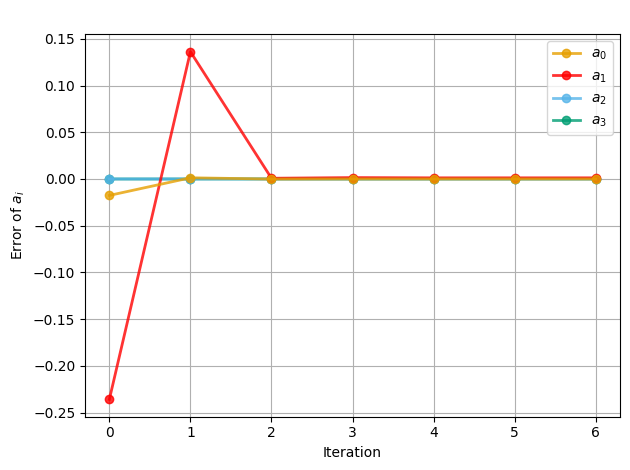}}
	\caption{Error of MFPI Algorithm on IEEE 14-bus.}
	\label{fig:mfpi}
\end{figure}

\vspace{15pt}
\section{Conclusion and Future Work}\label{sec:conc}
In this paper, two attack mechanisms leveraging LMPs and power generation data are proposed to infer the coefficients of cost functions in power grids. In the first scenario, a closed-form solution is proposed, while an MFPI algorithm is introduced with theoretical guarantees in the second scenario. In the future, this work will be extended in three directions: firstly, a stronger attack mechanism for inferring more private information (e.g., linear and fixed coefficients of cost functions in Scenario 2; secondly, research on leveraging the attack results for strategic bidding; and thirdly, deriving more general lower bound conditions to guarantee convergence.

\end{document}